\documentclass[10pt, twocolumn]{IEEEtran}

\usepackage{amsmath}
\usepackage{amssymb}    
\usepackage{amsfonts}
\usepackage[english]{babel}
\usepackage{cite} 
\usepackage{multirow}
\usepackage{stfloats}    

\usepackage{algorithm}   
\usepackage{algorithmic} 

\usepackage{caption} 
\usepackage{graphicx}
\usepackage{epsfig}
\usepackage{subfigure} 
\usepackage{tablefootnote}

\usepackage{url}

\usepackage{accents}
\makeatletter
\def\widebar{\accentset{{\cc@style\underline{\mskip10mu}}}}
\def\Widebar{\accentset{{\cc@style\underline{\mskip13mu}}}}
\makeatother

\newtheorem{theorem}{Theorem}

\newtheorem{definition}{Definition}
\newtheorem{proposition}{Proposition}
\newtheorem{corollary}{Corollary}

\begin{document}

\captionsetup[figure]{labelfont={ }, name={Fig.}, labelsep=period} 
\pagestyle{empty}

\title{Uplink Age of Information of Unilaterally Powered Two-way Data Exchanging Systems}
\author{\authorblockN{Yunquan Dong\authorrefmark{1},
                                           Zhengchuan Chen\authorrefmark{2},
                                           and Pingyi Fan\authorrefmark{3} \\
\authorblockA{ \normalsize
\authorrefmark{1}School of Electronic and Information Engineering, \\
                                 Nanjing University of Information Science \& Technology, Nanjing, China \\
\authorrefmark{2}College of Communication Engineering, Chongqing University, Chongqing, China \\
\authorrefmark{3}
                                    Department of Electronic Engineering, Tsinghua University, Beijing, China\\
yunquandong@nuist.edu.cn, czc@cqu.edu.cn, fpy@tsinghua.edu.cn
}
}
}

\maketitle
\thispagestyle{empty}

\vspace{-15mm}
\begin{abstract}
We consider a two-way data exchanging system where a master node transfers energy and data packets to a slave node alternatively.
    The slave node harvests the transferred energy and performs information transmission as long as it has sufficient energy for current block, i.e., according to the best-effort policy.
We examine the freshness of the received packets at the master node in terms of \textit{age of information} (AoI), which is defined as the time elapsed after the generation of the latest received packet.
    We derive average uplink AoI  and  uplink data rate as functions of downlink data rate in closed form.
The obtained results illustrate the performance limit of the unilaterally powered two-way data exchanging system in terms of timeliness and efficiency.
    The results also specify the achievable tradeoff between the data rates of the two-way data exchanging system.

\end{abstract}


\section{Introduction}
The recent prevalence of Internet-of-Things has spawned a plethora of real-time services that require timely data/update exchange~\cite{IoT-2017}.
     In vehicular networks~\cite{Vnet-1-2011, Vnet-2-2011}, for example, vehicles need to share their status information (e.g., position, speed, acceleration) timely to ensure safety.
For these scenarios, neither traditional delay nor traditional throughput is suitable~\cite{Kam-AoI-2015}.
    To be specific, when delay is small, the received data may not always be fresh if the transmissions are very infrequent;
when throughput is large, the received data may also be not fresh if the data undergo large queueing delay~\cite{Dong-Que-2014,Dong-2013-WCSP,Dong-DTL-2012}.
    To convey the freshness of the received information, therefore, a new metric was proposed in~\cite{Yates-2012-age}, i.e., \textit{age of information (AoI)} .

AoI is defined as the elapsed time since the generation of the latest received update~\cite{Yates-2012-age}, i.e., the age of the newest update at the receiver.
    Since AoI is closely related to queueing theory, it has been studied in various queueing systems, e.g., \textit{M/M/1, M/D/1} and \textit{D/M/1}~\cite{Yates-2012-age}, and under several serving disciplines, e.g., first-come-first-served (FCFS)~\cite{Yates-2012-age,Inoue-2017-distri}~and last-generate-first-served (LGFS)~\cite{Sun-2016-mlt-sver}.
The \textit{zero-wait} policy where a new update is served immediately after the completion of previous update was also investigated in~\cite{Sun-2016-zero-wait}.
    Moreover, the authors of~\cite{Yates-2017-erasure} studied the average AoI of transmitting $k$-symbol updates over an erasure channel. 
Although all of these serving disciplines can find many suitable applications, neither of them can be generally optimal.
    This has motivated many studies which minimize the AoI of the system by  scheduling the transmission of updates.
For example, the authors of ~\cite{Ephremides-2016-manage} discussed some protocols in which any arriving updates seeing a busy server or more than one waiting updates would be discarded.
    The method of replacing the waiting updates with newly arriving updates was studied in ~\cite{Ephremides-2016-manage}~\cite{Milcom-2016-erasure}.
The AoI of energy harvesting powered systems has also attracted many attentions.
    Due to the randomness of the energy harvesting process, energy buffers are needed to store the harvested energy.
To this end, the authors of~\cite{Yates-2012-lazy-eh, Yangjing-2017-eh} investigated how buffer size affects the average AoI of the system; the optimal threshold of remaining energy to trigger a new update has been found in ~\cite{Tan-2017-threshold}.
    Moreover, AoI was also investigated in  multi-source~\cite{Yate-2016-multisource}, multi-class~\cite{Huang-2015-multiclass}, multi-hop~\cite{Sun-2017-multihop} scenarios.

    In this paper, we consider the freshness of the uplink transmission from the slave node to the master node in a two-way data exchanging system, as shown in Fig. \ref{fig:net_model}.
We assume that only the master node has a constant energy supply.
    When the master node is not transmitting data, it transfers energy to the slave node using wireless power transfer~\cite{bio-2012-WPT}.
Using the harvested energy, the slave node can then transmit its own data to the master node.
    Since the freshness and effectiveness of uplink transmission are jointly constrained by the information transmission capability of the uplink channel and the energy transfer capability of the downlink channel, the corresponding performance analysis is more complicated and more meaningful.
We also would like to mention that the considered two-way data exchanging model can find many applications in implantable biomedical systems, device-to-device communications, and swarm robotics communications.

In these time-sensitive applications, the transmission time to deliver data packets over downlink and uplink channels cannot be neglected.
    In  this paper, therefore,  we model the transmission time of packets as the \textit{service time} of the data exchanging system.
We then investigate the transmission capability of uplink channel in terms of average uplink AoI and  uplink data rate, which are constrained by the limited energy at the slave node.
We prove that average uplink AoI approaches some constant as downlink data rate $p$ goes to zero and goes to infinity gradually when $p$ is increased.
    We also present uplink data rate $q$ as a function of $p$ in closed form.
We show that as $p$ goes to zero, uplink data rate $q$ is a constant; when $p$ is increased,   uplink data rate goes to zero gradually.
    Since downlink data rate and uplink data rate cannot be optimized at the same time, the obtained results also present the best achievable tradeoff between them.

This paper is organized as follows.
   Section~\ref{sec:2_model} presents the system model and the AoI model.
In Section~\ref{sec:4_aoi_ul}, we derive  average uplink AoI and uplink data rate as functions of downlink data rate $p$, as well as their asymptotic behaviors.
Finally, numerical results are provided in Section~\ref{sec5:simulation} and  our work is concluded in Section~\ref{sec6:conclusion}.

\section{System Model}\label{sec:2_model}

\begin{figure}[!t]
\centering
\includegraphics[width=3.1in]{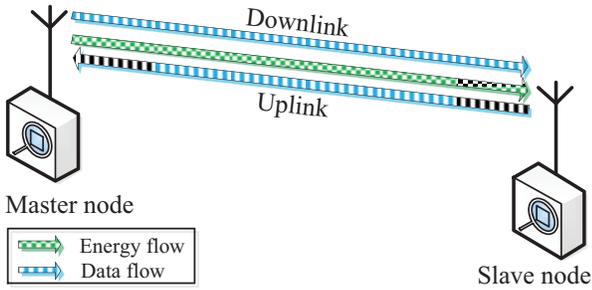}
\caption{Two-way data exchanging system with wireless power transfer. } \label{fig:net_model}
\end{figure}

We consider a two-way data exchanging system as shown in Fig.  \ref{fig:net_model}, where a master node and a slave node exchange their collected data through block fading channels.
    We assume that time is discrete and block length is $T_\text{B}$.
We also assume that the distance between the master node and the slave node is $1$ meter.

\subsection{Channel Model}
    We assume that the channels in both directions suffer from block Rayleigh-fading and additive white Gaussian noise (AWGN), i.e., the power gain $\gamma_n$  follows exponential distribution
                                  \begin{equation*} \label{eq:gama_pdf}
                                        f_\gamma(x) = \lambda e^{-\lambda x}.
                                    \end{equation*}

For both downlink and uplink transmissions, let $P_\text{t}$ be the transmit power, $\texttt{W}$ be the limited system bandwidth, and $\texttt{N}_0$ be the noise spectrum density.
    We further assume that $P_\text{t}$ is small and the received  signal-to-noise-ratio (SNR) is much smaller than unity\footnote{With some minor changes, our analysis can also be used for scenarios with general SNR and general fading models. The obtained results, however, might be more complicated in expression.}.
Hence, the amount of information that can be transmitted in a block would be
\begin{equation*} \label{rt:b_n}
    b_n = T_\text{B}\texttt{W} \log \left(1+ \tfrac{\gamma_n P_\text{t}}{\texttt{WN}_0} \right) \approx \tfrac{\gamma_n P_\text{t}T_\text{B}}{\texttt{N}_0}.
\end{equation*}
It is clear that $b_n$ also follows exponential distribution.
    In addition, we assume that downlink and uplink transmissions use different frequency bands.

\subsection{Data Exchanging Model}
In each block, the master node collects new data and then generates a packet of $\ell$ nats with probability $p$, which is referred to as the  \textit{downlink data rate}.
    The generated packets will be transmitted to the slave node according to FCFS policy.
We denote the number of blocks used to deliver a packet as \textit{service time} $S$.
    According to \cite{Dong-2013-WCSP},  $S$ follows Poisson distribution\footnote{Note that we have $j\geq 1$ in our case, which is slightly different from Poisson distribution. In addition, we say the service time is $S=1$ even if its completion time is less than $T_\text{B}$. }
\begin{equation} \label{rt:pk}
    p_j^\text{s} = \Pr\left\{\sum_{i=1}^{j-1} b_i < \ell, \sum_{i=1}^{j} b_i \geq \ell\right\} = \tfrac{\theta^{j-1}}{(j-1)!} e^{-\theta},
\end{equation}
where $j=1,2,\cdots$, and
\begin{equation} \label{df:theta}
    \theta=\tfrac{\lambda\texttt{N}_0\ell}{P_\text{t}T_\text{B}}.
\end{equation}

    The probability generating function (PGF) and the first two order moments of $S$ are, respectively, given by
\begin{eqnarray}
    \label{rt:Sk_pgf}                  G_\text{S}(z) \hspace{-2.75mm} &=& \hspace{-2.75mm} \mathbb{E} \left(z^{S}\right) = ze^{\theta(z-1)}, \\
    \label{rt:Sk_mean}              \mathbb{E}(S) \hspace{-2.75mm} &=& \hspace{-2.75mm} \lim_{z\rightarrow1^-} G'_\text{S}(z) = 1+\theta, \\
    \label{rt:Sk_2_moment}    \mathbb{E}(S^2) \hspace{-2.75mm} &=& \hspace{-2.75mm} \lim_{z\rightarrow1^-} G''_\text{S}(z) + G'_\text{S}(z) = \theta^2+3\theta +1.
\end{eqnarray}

A period when the master node is busy in transmitting packets is referred to as a \textit{downlink busy period} and a period when the data queue of the master node is empty and no new packet arrives is referred to as a \textit{downlink idle period}.
    In  downlink idle periods, the master node transfers energy to the slave node.
Let $\eta$ be the efficiency of downlink energy transfer, the energy received by the slave node would be
\begin{equation} \label{modl:wpt}
    E_n = \eta \gamma_n P_\text{t}T_\text{B}.
\end{equation}

We assume that the node tries its best to deliver its collected information and generates a data packet immediately after the completion of the previous packet, namely, according to the best-effort policy.
    Moreover, the slave node performs transmission to the master node as long as its remaining energy is no less than $P_\text{t}T_\text{B}$.

\subsection{Age of Information} \label{susec:2_3}
\begin{definition}
    In block $n$, uplink age of information  is the difference between  $n$ and the generation time $U(n)$ of the latest received packet at the master node:
    \begin{equation} \nonumber
        \Delta(n) = n - U(n).
    \end{equation}
\end{definition}

Fig. \ref{fig:aoi_ul} presents a sample variation of  uplink AoI with initial age $\Delta_0$.
     The packets are generated at arrival epochs $n_k$ and are completely transmitted at departure epochs $n'_k$.
We denote the time that an uplink packet $k$ stays in the system as  \textit{uplink system time} $T_k$.
    Note that if  downlink power transfer is weak, an uplink packet would take a long period to complete, which may covers several downlink idle periods and busy periods.
Thus,  $T_k$ includes the service time of the packet, the time for harvesting energy, and the time waiting for downlink power transfer (if  any).

\begin{figure}[!t]
\centering
\includegraphics[width=2.95in]{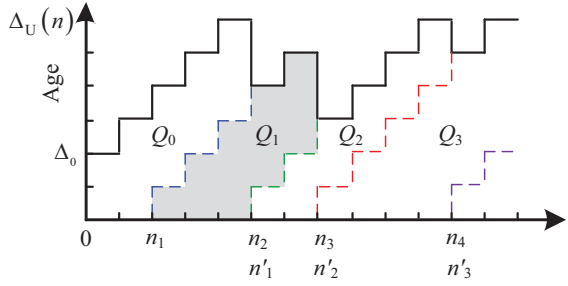}
\caption{Sample path of uplink AoI $\Delta(n)$ (the upper envelope in bold). } \label{fig:aoi_ul}
\end{figure}

It is observed that AoI increases linearly in time and is reset to a smaller value (the age of a newer packet) at the end of departure blocks.
    Over a period of $N$ blocks where $K$ uplink packets are transmitted, the average uplink AoI is defined as
    \begin{equation} \nonumber
        \Delta_{N} = \frac 1N \sum_{n=1}^{N} \Delta(n).
    \end{equation}
Starting from the first block, the area under $\Delta_{}(n)$ can be seen as the concatenation of areas $Q_0, Q_1, \cdots$, and the triangular-like area of width $T_K$.
    Thus,  average uplink AoI can be rewritten as
    \begin{equation} \label{df:av_aoi_ul}
        \widebar\Delta = \lim_{N\rightarrow\infty}\frac 1N \left( Q_0 +\sum_{k=1}^{K-1} Q_k + \frac 12 T_K(T_K+1) \right) .
    \end{equation}

Given the downlink data rate $p$, we denote the achievable uplink data rate as $q(p)$, which is given by
\begin{equation} \label{df:ul_rate}
    q(p) = \lim_{N\rightarrow\infty} \frac{K}{N}.
\end{equation}

\section{Uplink Age of Information} \label{sec:4_aoi_ul}
    In this section, we first investigate the statistic property of the energy harvesting process and the uplink service time, and then derive  average uplink AoI in closed form.

\subsection{Energy Harvesting Process} \label{subsuc:4_A}
Note that  energy transfer efficiency is always smaller than unity.
    In each block, therefore, the harvested energy at the slave node can support one block of transmission at most.
In most cases, therefore, there is no energy left during downlink busy periods at the slave node, which implies that uplink transmissions can only occur during downlink idle periods.
    In a period containing $j$ downlink idle blocks, we denote $e_j$ as the harvested energy. That is,
\begin{equation*}
    e_j=\sum_{i=1}^{j} E_{m_i}=\eta P_\text{t}T_\text{B} \sum_{i=1}^{j} \gamma_{m_i},
\end{equation*}
where $m_i$ is the index of the $i$-th block of the period and $\gamma_i$ is the exponentially distributed power gain.
    It is clear that $e_j$ follows $\text{Erlang}\big(j,\tfrac{\lambda}{\eta P_\text{t}T_\text{B}}\big)$ distribution and the corresponding probability density function is given by
    \begin{equation} \label{eq:f_ej}
        f_{e_j}(x) = \tfrac{\mu^jx^{j-1}e^{-\mu x}}{(j-1)!},
    \end{equation}
    where $\mu=\tfrac{\lambda}{\eta P_\text{t}T_\text{B}}$.

Let $\tau_\text{H}$ be the number of downlink idle blocks for the slave node to accumulate sufficient energy to perform a block of transmission and  $e_r$ be the remaining energy after the previous transmission.
    We have $\Pr\{ \tau_\text{H} = 0\} = \Pr\{ e_r > P_\text{t}T_\text{B} \}$ and
    \begin{equation*} \label{rt:harvesting_time_ul}
        \Pr\{ \tau_\text{H} = j\} = \Pr\{ e_r + e_{j-1}<P_\text{t}T_\text{B}, e_r +  e_{j}\geq P_\text{t}T_\text{B} \},
    \end{equation*}
    for  $j=1,2,\cdots$.
Thus, the actual time to perform the uplink transmission in the $i$-th block is $s_i=\max\{ 1, \tau_{\text{H}i} \}$.

    Since both energy transfer efficiency $\eta$ and average channel power gain $\mathbb{E}[\gamma]=\frac1\lambda$ are quite small in general, the harvested energy in each block is not enough to perform a block of transmission.
Thus, $\tau_\text{H}$ would be larger than 1 almost surely.
    In this case, $e_r$ would be the the remaining part of $E_m$, which is the energy harvested in each block.
According to \eqref{modl:wpt}, $E_m$ follows exponential distribution, which is memoryless.
    Thus, $e_r$ follows the same distribution as $E_m$ and we have,
        \begin{eqnarray}
         \nonumber          \hspace{-2.75mm}&&\hspace{-2.75mm}  \Pr\{ \tau_\text{H} = j\} = \Pr\{ e_{j}<P_\text{t}T_\text{B},  e_{j+1}\geq P_\text{t}T_\text{B} \} \\
         \label{rt:p_tau_h}         \hspace{-2.75mm}&&\hspace{-2.75mm} = \int_0^{P_\text{t}T_\text{B}} f_{e_{j}}(x) \text{d}x \int_{P_\text{t}T_\text{B}-x}^\infty f_{e_{1}}(y) \text{d}y
                                                                                        = \tfrac{(\frac{\lambda}{\eta})^{j}}{j!} e^{\frac{-\lambda}{\eta}}.
    \end{eqnarray}

Since  uplink transmissions suffer the same fading and use the same transmit power as downlink transmissions,  the required number of blocks to deliver an uplink packet also follows distribution law \eqref{rt:pk}.
    However, the slave node often needs to wait for some time to accumulate enough energy to deliver a packet.
Thus, the \textit{uplink service time}  $S_\text{UL}$, which includes the actual service time for transmitting the packet (i.e., $S$), and the time required for harvesting and accumulating energy, must be larger than the downlink service time for the master node to deliver a downlink packet.
In fact, we have
    \begin{equation} \label{rt:service_time_ul}
        S_\text{UL} = \sum_{i=1}^{S} s_i.
    \end{equation}

In particular, the moments of $S_\text{UL}$ are given by the following proposition.
\begin{proposition} \label{prop4:service_time_ul}
    The first-two order moments of the uplink service time $S_\text{U}$ are, respectively,  given by
    \begin{eqnarray}
        \label{rt:e_ul_service_time_1}   \mathbb{E}(S_\text{UL})      \hspace{-2.75mm} &=& \hspace{-2.75mm} \Big( \tfrac{\lambda}{\eta} + e^{-\frac{\lambda}{\eta}} \Big) (1+\theta), \\
        \nonumber                                       \mathbb{E}(S_\text{UL}^2) \hspace{-2.75mm} &=& \hspace{-2.75mm}  \Big( \tfrac{\lambda}{\eta} + e^{-\frac{\lambda}{\eta}} \Big) ^2(\theta^2+2\theta)
                                                                                   +  \Big( \tfrac{\lambda^2}{\eta^2} + \tfrac{\lambda}{\eta} + e^{-\frac{\lambda}{\eta}} \Big) (1+\theta).
    \end{eqnarray}
\end{proposition}

\begin{proof}
    See Appendix \ref{prf4:service_time_ul}.
\end{proof}


\subsection{Uplink System Time} \label{subsuc:4_B}
Note that  two busy periods would appear uninterruptedly if a new downlink packet is generated immediately after the end of a busy period.
    In this case, we say that there is a downlink idle period of zero length between the two consecutive downlink busy period, i.e., $I_{\text{D}}=0$.
Note also that the distribution of downlink idle period is given by $\Pr\{ I_{\text{D}}=j \} =(1-p)^j p ,  j=0,1,\cdots$.
    We denote the number of downlink busy periods coming uninterruptedly as $F$, then the distribution and the first two order moments of $F$ can be, respectively, given by
\begin{eqnarray}
       \label{rt:distrib_Fi}   p_j^\text{f} \hspace{-2.75mm}&=&\hspace{-2.75mm} \Pr\{ F=j \} = p^j (1-p), ~~~ j=0,1,2,\cdots, \\
       \label{rt:F_1st}          \mathbb{E}(F) \hspace{-2.75mm}&=&\hspace{-2.75mm}  \tfrac{p}{1-p}, \\
       \label{rt:F_2ed}        \mathbb{E}(F^2) \hspace{-2.75mm}&=&\hspace{-2.75mm}  \tfrac{p(1+p)}{(1-p)^2}.
\end{eqnarray}

    Since there might exist one or more downlink busy periods before each of the blocks of uplink service time $S_\text{UL}$, the uplink system time of a packet can be expressed as
\begin{equation} \label{df:X_be}
    T = S_{\text{UL}} + \sum_{i=1}^{S_{\text{UL}}} \sum_{j=1}^{F_i} B_{\text{D}j}.
\end{equation}

In particular, the first two order moments of downlink busy period $B_{\text{D}}$ are given by the following proposition.
\begin{proposition} \label{prop:5_busy_dl}
    The first-order and second-order moments of downlink  busy period are, respectively, given by
    \begin{eqnarray}
        \label{rt:busy_dl_e} \mathbb{E}(B_{\text{D}}) \hspace{-2.75mm} &=& \hspace{-2.75mm} \tfrac{1+\theta}{1-p-\theta p}, \\
        \label{rt:busy_dl_e2} \mathbb{E}(B_{\text{D}}^2) \hspace{-2.75mm} &=& \hspace{-2.75mm} \tfrac{ (1+\theta)^2 (1-p^2-\theta p^2) +\theta }{(1-p-\theta p)^3}.
    \end{eqnarray}
\end{proposition}

\begin{proof}
    The proof follows the same outline as that of \textit{Proposition} 2 in \cite{Dong-Que-2014}.
\end{proof}

Based on aforementioned analysis and \textit{Proposition} \ref{prop:5_busy_dl}, we have the following proposition on uplink system time $T$.
\begin{proposition} \label{prop:6_system_ul}
    The first-two order moments of uplink  system time are, respectively, given by,
    \begin{eqnarray}
        \label{rt:system_ul_e}           \mathbb{E}(T) \hspace{-2.75mm} &=& \hspace{-2.75mm}  \mathbb{E}(S_{\text{U}}) \big( 1+ \mathbb{E}(F)\mathbb{E}(B_{\text{D}}) \big), \\
        \nonumber         \mathbb{E}(T^2) \hspace{-3mm}&=&\hspace{-3mm} \mathbb{E}(S_{\text{U}}^2) \Big( 1+ \mathbb{E}(F) \mathbb{E}(B_{\text{D}}) \Big)^2 \\
     \nonumber    \hspace{-3mm}&&\hspace{-3mm}  +\mathbb{E}(S_{\text{U}})
                                                                                            \Big( \mathbb{E}(F) \big( \mathbb{E}(B_{\text{D}}^2) - \mathbb{E}^2(B_{\text{D}})  \big) \\
     \nonumber    \hspace{-3mm}&&\hspace{-3mm}  ~~~~~~~~~~~+ \big( \mathbb{E}(F^2) -\mathbb{E}^2(F) \big) \mathbb{E}^2 (B_{\text{D}}) \Big).
    \end{eqnarray}
\end{proposition}

\begin{proof}
    See Appendix \ref{prf:prop6_system}.
\end{proof}

\begin{figure*}[htp]   

\hspace{-6 mm}
    \begin{tabular}{cc}
    \subfigure[Average uplink AoI]
    {
    \begin{minipage}[t]{0.5\textwidth}
    \centering
    {\includegraphics[width = 3.2in] {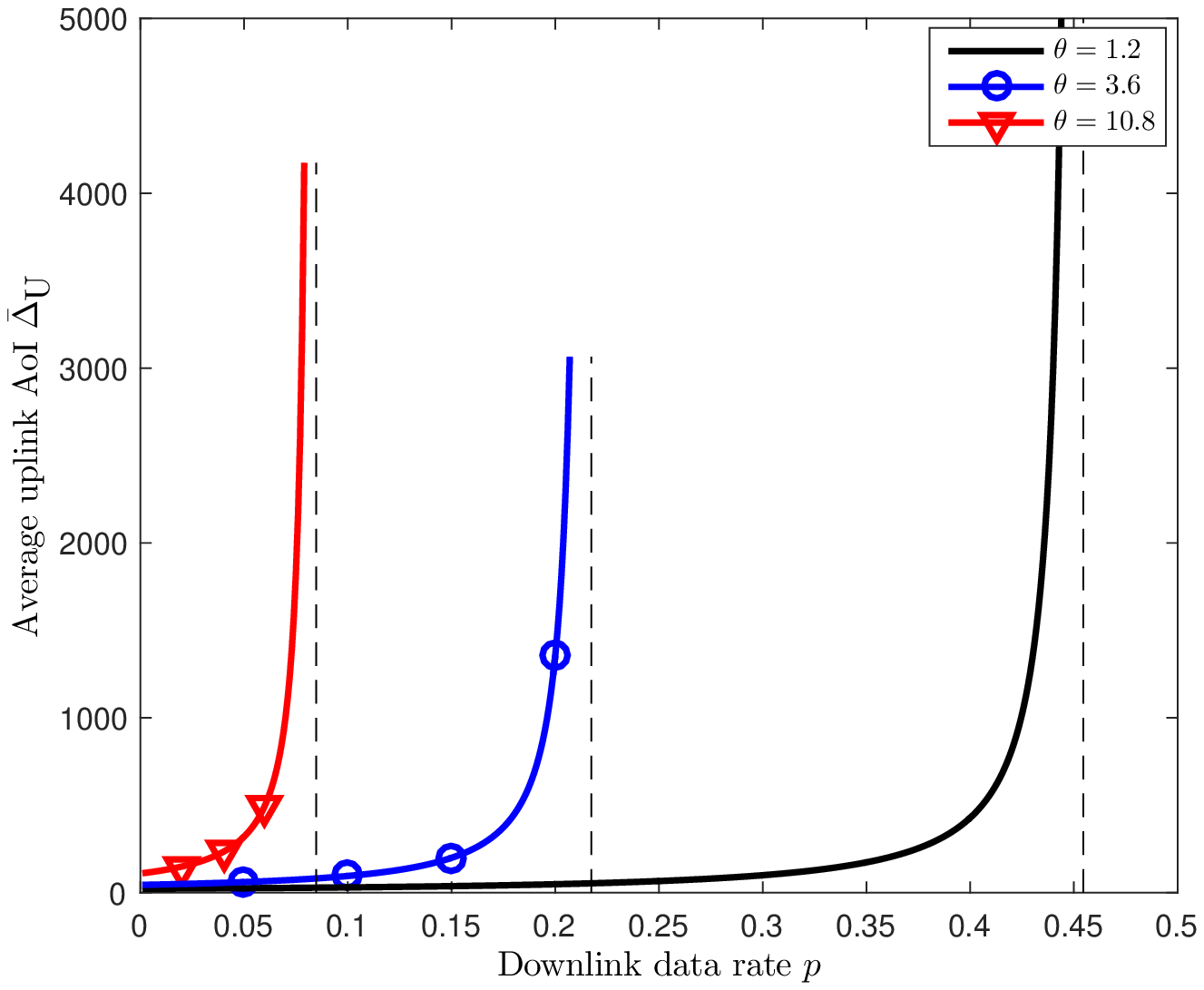} \label{fig:aoi_ul_sim}}
    \end{minipage}
    }

    \subfigure[Average uplink data rate $q$]
    {
    \begin{minipage}[t]{0.5\textwidth}
    \centering
    {\includegraphics[width = 3.2in] {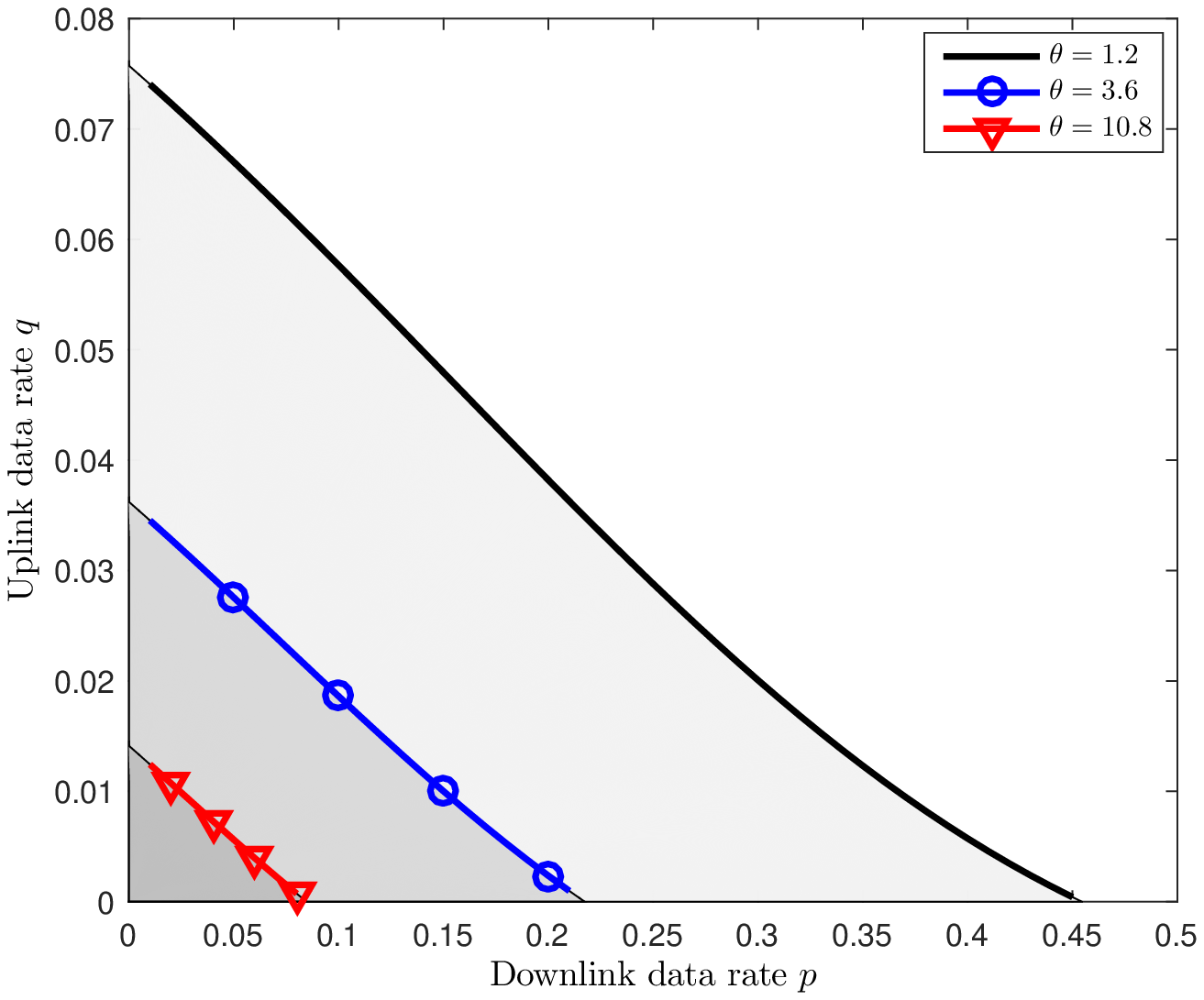} \label{fig:q_ul_sim}}
    \end{minipage}
    }

    \end{tabular}

\caption{Average AoI and data rate of the two-way data exchanging. } \label{fig:aoi}
\end{figure*}

\subsection{Average Uplink AoI}
Based on previous discussions, we have the following theorem on average uplink AoI and uplink data rate.

\begin{theorem} \label{thm:ul}
    If $(\theta+1)p<1$,  average uplink AoI would be 
        \begin{eqnarray} \nonumber
            \hspace{-7mm}   \widebar\Delta \hspace{-3.3mm} &= & \hspace{-3.5mm}    \tfrac {1-p+p^2+\theta p^2}   {2(1-p)(1-p-\theta p)}
                                                              \left(  \tfrac{ \frac{\lambda^2}{\eta^2} + \frac\lambda\eta+e^{-\frac\lambda\eta} }  {\frac\lambda\eta+e^{-\frac\lambda\eta}}
                                                                       + \big(\tfrac\lambda\eta+e^{-\frac\lambda\eta}\big) \tfrac{ 3\theta^2+6\theta+2 }  {1+\theta}
                                                              \right)   \\
            \label{rt:thm_2_ulaoi}                    \hspace{-2.75mm} && \hspace{-2.75mm}  + \tfrac{  p (1+\theta)^2 -p^3(2-p)(1+\theta)^3 +\theta p(1-p)  }
                                                              {2(1-p)(1-p-\theta p)^2(1-p+p^2+\theta p^2)}  +\tfrac 12.
    \end{eqnarray}
    the achievable uplink data rate is given by
        \begin{equation} \label{rt:thm_2_ulq}
            q(p)= \tfrac {(1-p)(1-p-\theta p)}  { \big( \frac{\lambda}{\eta} + e^{-\frac{\lambda}{\eta} } \big) (1+\theta)  (1-p+p^2+\theta p^2)},
        \end{equation}
    where $\theta=\tfrac{\lambda\texttt{N}_0\ell}{P_\text{t}T_\text{B}}$. Otherwise, $\widebar\Delta$ would be infinitely large and $q(p)$ would be zero.
\end{theorem}

\begin{proof}
\textit{Theorem} \ref{thm:ul} follows directly from \textit{Proposition} \ref{prop:5_busy_dl} and \ref{prop:6_system_ul}.
    For more details, please see Appendix \ref{prf:thm_2_ul}.
\end{proof}

To obtain more insights, we further investigate the following two special cases.

\begin{corollary} \label{coro:aoi_ul_3}
    In the case $(\theta+1)p \approx o$ where $o$ is an infinitesimal, average uplink AoI and uplink data  rate are, respectively, given by
    \begin{eqnarray*} \label{rt:coro3_aoi_ul}
        \widebar\Delta \hspace{-2.75mm} &=& \hspace{-2.75mm} \frac 12
                                                              \left(  \tfrac{ \frac{\lambda^2}{\eta^2} + \frac\lambda\eta + e^{-\frac{\lambda}{\eta} } }  {\frac\lambda\eta+e^{-\frac{\lambda}{\eta} }}
                                                                       +  \tfrac{ \big(\tfrac\lambda\eta+e^{-\frac{\lambda}{\eta} } \big) (3\theta^2+6\theta+2) }  {1+\theta}
                                                              \right)     +\frac 12, \\
        q(p) \hspace{-2.75mm} &=& \hspace{-2.75mm} \tfrac {1}  { \big( \frac{\lambda}{\eta} + e^{-\frac{\lambda}{\eta} } \big) (1+\theta)  }.
    \end{eqnarray*}
\end{corollary}
\begin{proof}
    The corollary readily follows the approximation $p\approx0$  and equations \eqref{rt:thm_2_ulaoi} and \eqref{rt:thm_2_ulq}.
\end{proof}

The condition $(\theta+1)p\thickapprox0$ indicates that  downlink channel is lightly occupied by information transmission and thus much energy can be transferred to the slave node.
    Under the best-effort policy, therefore, both uplink AoI and uplink data rate would be finite, as shown in \textit{Corollary} \ref{coro:aoi_ul_3}.
Note that the key parameters in this case include the expected service time $\mathbb{E}(S)=1+\theta$, the expected channel power gain $\frac1\lambda$ and the energy transfer efficiency $\eta$.

\vspace{3mm}
\begin{corollary} \label{coro:aoi_ul_4}
    In the case $(\theta+1)p \approx 1-o$ where $o$ is an infinitesimal, average uplink AoI and uplink data rate are, respectively, given by
    \begin{eqnarray*} \label{rt:coro4_aoi_ul}
        \widebar\Delta \hspace{-2.75mm} &=& \hspace{-2.75mm} \left(  \tfrac{ \frac{\lambda^2}{\eta^2}
                                                                        + \frac\lambda\eta +e^{-\frac{\lambda}{\eta} } }  {\frac\lambda\eta+e^{-\frac{\lambda}{\eta} }}
                                                                       +  \tfrac{ \big(\tfrac\lambda\eta+e^{-\frac{\lambda}{\eta} }\big) (3\theta^2+6\theta+2) }  {1+\theta}
                                                                    \right)  \tfrac{1+\theta}{2\theta} \cdot \tfrac{1}{o} +\tfrac 12, \\
        q(p) \hspace{-2.75mm} &=& \hspace{-2.75mm}   \tfrac{\theta}  {(1+\theta)^2\big(\frac\lambda\eta+ e^{-\frac{\lambda}{\eta} } \big)} \cdot o.
    \end{eqnarray*}
\end{corollary}
\begin{proof}
    This corollary can be readily proved by using the approximation in equations \eqref{rt:thm_2_ulaoi} and \eqref{rt:thm_2_ulq}.
\end{proof}

When $(\theta+1)p$ approaches unity, the downlink channel is very busy and little energy can be harvested at the slave node.
    Thus, the uplink service time would be very large.
As a result, uplink AoI goes to infinity and uplink data rate goes to zero, as observed in \textit{Corollary} \ref{coro:aoi_ul_4}.

\section{Numerical Results}\label{sec5:simulation}
In this section, we investigate the AoI of the two-way data exchanging system via numerical results.
    We set the transmit power of both the master node and the slave node, for both transmitting information and transferring energy, as $P_\text{t}= 0.01$ W.
The system bandwidth is  $\texttt{W}=1$ MHz, the noise spectrum density (including noise figure, etc.) is  $\texttt{N}_0=4\times10^{-7}$.
    The Rayleigh channel parameter is $\lambda=3$,  the block length is $T_\text{B}=10^{-3}$ s, and the energy transfer efficiency is $\eta=0.5$.
For simplicity, we set the distance between the master node and the slave node to $d=1$ m.
    For pecket length, we consider the following three cases: $\ell=10$ nats, $\ell=30$ nats, and $\ell=90$ nats.
By using \eqref{df:theta}, the corresponding parameter $\theta$ can be obtained as $\theta=1.2$, $\theta=3.6$, and  $\theta=10.8$, respectively.
    Note that for any given $\theta$, the maximal downlink data rate enabling a stable queue at the master node is $p_{\max}=\tfrac{1}{1+\theta}$.

We present average uplink AoI in Fig. \ref{fig:aoi_ul_sim}.
    In general, average uplink AoI is  large, especially when downlink data rate $p$ is large.
On one hand, if $p$ is very small, the master node would transfer energy to the slave node for most of the time.
    In this case, the slave node seldom needs to  wait for harvesting energy so that uplink system time is determined only by energy transfer efficiency $\eta$ and actual service time $S$.
As $p$ approaches zero, therefore, average uplink AoI would converge to a constant, as shown in Fig. \ref{fig:aoi_ul_sim} and \textit{Corollary} \ref{coro:aoi_ul_3}.
    On the other hand, as $p$ approaches the maximal downlink data rate $p_\text{max}$, average uplink AoI goes to infinity, which is consistent with \textit{Corollary} \ref{coro:aoi_ul_4}.

We plot how uplink data rate $q$ (see \eqref{rt:thm_2_ulq}) varies when downlink data rate $p$ is changed in Fig. \ref{fig:q_ul_sim}.
    We observe that $q$ decreases rapidly when $p$ is increased.
In particular, $q$ reduces to zero as $p$ approaches $p_\text{max}$.
    Moreover, for each given $\theta$, the area under the curve can be regarded as the \textit{achievable region} of  data rate pair $(p,q)$.
That is, each point under the curve is achievable while the points above the curve are not.
    Since the power supply at the master node is the only energy source of the system, downlink data rate and uplink data rate cannot be optimized at the same time.
    Thus, the curves in Fig. \ref{fig:q_ul_sim} can also be regarded as best-achievable tradeoff between downlink data rate and uplink data rate.
To see this clearly, one may use a weighted-sum characterization of the system data rate.
    That is, a data rate pair $(p,q)$ is said to be optimal if it maximizes the  weighted sum data rate $wp+(1-w)q$, where $0\leq w\leq1$ is a priority constant.
Intuitively, the solution to this optimization problem can be obtained by searching the tangent point between line $wp+(1-w)q=c$ and the curves in Fig. \ref{fig:q_ul_sim}.

\section{Conclusion}\label{sec6:conclusion}
In this paper, we have studied the AoI of a two-way data exchanging system with a unique power supply at the master node.
    We obtained average uplink AoI in closed form and presented the corresponding asymptotic behavior.
Based on these results, we can also determine the achievable region of  downlink data rate and uplink data rate.
    It is clear that downlink performance is constrained  by the limited transmit power of the master node.
The uplink performance, however, is also affected by downlink data rate $p$.
    To be specific, average uplink AoI would be smaller and average uplink data rate would be larger if $p$ is decreased.
    Since the performance of downlink and uplink cannot be optimized at the same time, one needs to find the tradeoff between them under some criteria, e.g., weighted-min/max, as discussed in Subsection \ref{sec5:simulation} and Fig. \ref{fig:q_ul_sim}.
Thus, the obtained results have presented a full characterization of the data exchanging capability of the two-way system.
    Note that this work has focused on the FCFS serving discipline.
Considering the performance of the system under other serving disciplines and packet management is also very interesting and will be explored in our future work.

\appendix

\renewcommand{\theequation}{\thesection.\arabic{equation}}
\newcounter{mytempthcnt}
\setcounter{mytempthcnt}{\value{theorem}}
\setcounter{theorem}{2}
\addcontentsline{toc}{section}{Appendices}\markboth{APPENDICES}{}

\subsection{Proof of Proposition \ref{prop4:service_time_ul}} \label{prf4:service_time_ul}
\begin{proof}
    Since the actual time to perform a block of uplink transmission is $s=\max\{ 1, \tau_\text{H} \}$, we have $\Pr\{s=1\} = \Pr\{\tau_\text{H}=0\}  + \Pr\{\tau_\text{H}=1\} = \big(1+\frac\lambda\eta\big) e^{-\lambda\eta}$.

    Based on \eqref{rt:p_tau_h} and \eqref{rt:service_time_ul}, we have
        \begin{eqnarray*}
        \mathbb{E}(S_\text{U})  \hspace{-2.75mm}&=&\hspace{-2.75mm} \mathbb{E}(S) \mathbb{E}(s)= \Big( \tfrac{\lambda}{\eta} + e^{-\frac{\lambda}{\eta}} \Big) (1+\theta), \\
        \mathbb{E}(S_\text{U}^2)  \hspace{-2.75mm}&=&\hspace{-2.75mm} \mathbb{E} \left(\sum_{i=1}^S s_i^2 + \sum_{i=1}^S \sum_{j\neq i}^S s_is_j  \right)\\
        \hspace{-2.75mm}&=&\hspace{-2.75mm} \mathbb{E}(S) \mathbb{E}(s^2)  + \mathbb{E}(S^2-S) \mathbb{E}^2(s) \\
        \hspace{-2.75mm}&=&\hspace{-2.75mm}   \Big( \tfrac{\lambda}{\eta} + e^{-\frac{\lambda}{\eta}} \Big) ^2(\theta^2+2\theta)
                                                                                   +  \Big( \tfrac{\lambda^2}{\eta^2} + \tfrac{\lambda}{\eta} + e^{-\frac{\lambda}{\eta}} \Big) (1+\theta).
    \end{eqnarray*}
     This completes the proof of the proposition.
\end{proof}

\subsection{Proof of Proposition \ref{prop:6_system_ul}} \label{prf:prop6_system}
\begin{proof}
    According to the definition of uplink system time \eqref{df:X_be} and the independency among downlink busy periods,
    \begin{eqnarray}
        \nonumber        \mathbb{E}(T) \hspace{-2.75mm}&=&\hspace{-2.75mm} \mathbb{E}\left(S_\text{UL} + \sum_{i=1}^{S_\text{UL}} \sum_{j=1}^{F_i} B_{\text{D}j} \right) \\
         \label{dr:X_be_e}                                                           \hspace{-2.75mm}&=&\hspace{-2.75mm} \mathbb{E}(S_\text{UL} ) + \mathbb{E}(S_\text{UL} )\mathbb{E}(F )\mathbb{E}( B_{\text{D}} ),
    \end{eqnarray}
    where $\mathbb{E}(S_\text{UL} ), \mathbb{E}(F )$, and $\mathbb{E}( B_{\text{D}} )$ are given by \eqref{rt:e_ul_service_time_1}, \eqref{rt:F_1st}, and \eqref{rt:busy_dl_e}, respectively.

    Denote $y_i=\sum_{j=1}^{F_i} B_{\text{D}j}$, we then have
    \begin{eqnarray}
        \nonumber      \hspace{-10mm}&&\hspace{-3mm}  \mathbb{E}(T^2) = \mathbb{E}\left(\Big(S_\text{UL} + \sum_{i=1}^{S_\text{UL}} y_i \Big)^2\right) \\
        \label{dr:X_be_e21}         \hspace{-10mm}&&\hspace{-3mm} = \mathbb{E}(S_\text{UL}^2 ) \hspace{-0.5mm}+\hspace{-0.5mm}
                                                                                                                2\mathbb{E}(S_\text{UL}^2 )\mathbb{E}(F )\mathbb{E}( B_{\text{D}} )
                                                                                                                  \hspace{-0.5mm}+\hspace{-0.5mm}  \mathbb{E}\left(\Big(\sum_{i=1}^{S_\text{UL}} y_i \Big)^2 \right),
    \end{eqnarray}
    where the last term is given by
    \begin{eqnarray}
        \nonumber        \hspace{-3mm}&&\hspace{-3mm} \mathbb{E}\left(\Big(\sum_{i=1}^{S_\text{UL}} y_i \Big)^2 \right)
                                                                                                = \mathbb{E}\left(\sum_{i=1}^{S_\text{UL}} y_i^2 +  \sum_{i_1=1}^{S_\text{UL}}  \sum_{i_2\neq i_1}^{S_\text{UL}} y_{i_1}y_{i_2} \right) \\
        \nonumber        \hspace{-2.75mm}&&\hspace{-2.75mm} = \mathbb{E}(S_\text{UL} ) \mathbb{E}(F ) \mathbb{E}(B_\text{D}^2 )
                                                                                                                                                                                                                + \mathbb{E}(S_\text{UL} ) \mathbb{E}(F^2-F ) \mathbb{E}^2(B_\text{D} )  \\
       \label{dr:X_be_e23}                                            \hspace{-2.75mm}& &\hspace{-2.75mm}    +\mathbb{E}(S_\text{UL}^2-S_\text{UL} ) \mathbb{E}^2(F ) \mathbb{E}^2( B_{\text{D}} ).
    \end{eqnarray}

    The proof of Proposition \ref{prop:6_system_ul} is readily completed by combing results \eqref{dr:X_be_e}--\eqref{dr:X_be_e23}.

\end{proof}

\subsection{Proof of Theorem \ref{thm:ul}}  \label{prf:thm_2_ul}
\begin{proof}
In the downlink transission, note that $p$ is the data rate  and $\mathbb{E}(S_k)=1+\theta$ is the average service time.
    Thus, the queue length at the master node would be infinitely large if $p(\theta+1) \geq 1$.
    That is, the master node will be always in the busy period and no energy can be transferred to the slave node.
In this case, the average uplink AoI would be infinitely large and the uplink data rate would be zero.
    Next, we consider the case of $p(\theta+1)<1$.

According to the definition of uplink data rate \eqref{df:ul_rate}, we have
\begin{eqnarray}
     \nonumber              q(p) \hspace{-2.75mm}&=&\hspace{-2.75mm} \lim_{N\rightarrow\infty} \frac{K}{N}=\frac{1}{\mathbb{E}(T) } \\
     \label{rt:ul_rate}      \hspace{-2.75mm}&=&\hspace{-2.75mm}\tfrac{(1-p)(1-p-\theta p)}
                                                                                                                                 { \big( \frac{\lambda}{\eta} + e^{-\frac{\lambda}{\eta}} \big) (1+\theta)(1-p+p^2+\theta p^2)},
\end{eqnarray}
where $\mathbb{E}(T)$ is given by \eqref{rt:system_ul_e}.

In the definition of $\widebar\Delta$  (cf. \eqref{df:av_aoi_ul}), we note that $Q_0$ is finite and $T_k$ is finite in probability. Thus, the average uplink AoI can be rewritten as
\begin{eqnarray}
        \nonumber   \widebar\Delta \hspace{-3mm} &=& \hspace{-3mm}  \lim_{N\rightarrow\infty} \frac 1N\sum_{k=1}^{K-1} Q_k \\
        \nonumber                                      \hspace{-3mm} &=& \hspace{-3mm} \lim_{N\rightarrow\infty} \frac {K-1}{N} \frac {1} {K-1} \sum_{k=1}^{K-1} Q_k \\
        \label{apx:avg_aoi_ul_1}             \hspace{-3mm} &=& \hspace{-3mm} q \mathbb{E} (Q_k).
    \end{eqnarray}

Note that the average of area $Q_k$ is given by
\begin{eqnarray}
    \nonumber    \mathbb{E} (Q_k) \hspace{-2.75mm} &=& \hspace{-2.75mm} \mathbb{E} (T_k T_{k+1}  + \tfrac12T_{k}(T_{k}+1) ) \\
    \label{apx:avg_aoi_ul_2}            \hspace{-2.75mm} &=& \hspace{-2.75mm} \mathbb{E}^2 ( T) + \tfrac12\mathbb{E} ( T^2) + \tfrac12 \mathbb{E} ( T).
\end{eqnarray}

Using the results in \textit{Proposition} \ref{prop:6_system_ul} and \eqref{rt:ul_rate}--\eqref{apx:avg_aoi_ul_2}, the average uplink AoI can be rewritten as
\begin{eqnarray*}
        \widebar\Delta \hspace{-2.75mm} &=& \hspace{-2.75mm}\mathbb{E} ( T) + \frac12+ \frac12 \frac{\mathbb{E} ( T^2) }
                                                                                                                                                            { \mathbb{E} ( T) } \\
                                                              \hspace{-2.75mm} &=& \hspace{-2.75mm} \tfrac {1-p+p^2+\theta p^2}   {2(1-p)(1-p-\theta p)}
                                                              \left(  \tfrac{ \frac{\lambda^2}{\eta^2} + \frac\lambda\eta+e^{-\frac\lambda\eta} }  {\frac\lambda\eta+e^{-\frac\lambda\eta}}
                                                                       + \big(\tfrac\lambda\eta+e^{-\frac\lambda\eta}\big) \tfrac{ 3\theta^2+6\theta+2 }  {1+\theta}
                                                              \right)   \\
                                                              \hspace{-2.75mm} && \hspace{-2.75mm}  + \tfrac{  p (1+\theta)^2 -p^3(2-p)(1+\theta)^3 +\theta p(1-p)  }
                                                              {(1-p)(1-p-\theta p)^2(1-p+p^2+\theta p^2)}  +\tfrac 12.
\end{eqnarray*}

Thus, the proof of \textit{Theorem} \ref{thm:ul} is completed.
\end{proof}

\section*{Acknowledgement}
This work was supported in parts by the National Natural Science Foundation of China (NSFC) under Grant 61701247,  the Jiangsu Provincial Natural Science Research Project under Grant 17KJB510035,
and the Startup Foundation for Introducing Talent of NUIST under Grant 2243141701008;
    Pingyi Fan's work was supported in parts by the China Major State Basic Research Development Program (973 Program) No. 2012CB316100(2) and the National Natural Science Foundation of China (NSFC) No. 61171064 and NSFC No. 61621091.

\small{
\bibliographystyle{IEEEtran}

}
\end{document}